\documentclass[sigconf,nonacm]{acmart}

\usepackage{amsmath}
\usepackage{booktabs}
\usepackage{array}
\usepackage{tikz}
\usetikzlibrary{arrows.meta,positioning,calc,patterns}

\newcommand{\Kset}{\mathcal{K}}
\newcommand{\argmax}{\operatorname*{arg\,max}}

\newcommand{\admPreRows}{1,400}

\newcommand{\prGpuShareHi}{76.1}
\newcommand{\prGpuShareLo}{12.4}
\newcommand{\prSerialBillHi}{8.7}
\newcommand{\prSerialBillLo}{7.7}
\newcommand{\prSerialCostHi}{2.0}
\newcommand{\prSerialCostLo}{1.3}

\newcommand{\rtAdditionalCorrect}{19}
\newcommand{\rtAdditionalWrong}{8}
\newcommand{\rtBuyFixedConst}{0.1225}
\newcommand{\rtBuyFixedError}{0.1600}
\newcommand{\rtBuyFixedOutside}{2.1200}
\newcommand{\rtBuyFixedWait}{0.9678}
\newcommand{\rtBuyMeteredConst}{0.1514}
\newcommand{\rtBuyMeteredError}{0.2319}
\newcommand{\rtBuyMeteredOutside}{1.5362}
\newcommand{\rtBuyMeteredWait}{1.5183}
\newcommand{\rtCoderErrorHigh}{1.8}
\newcommand{\rtCoderErrorLow}{0.5}
\newcommand{\rtIntegrationCloudLoss}{8.3}
\newcommand{\rtIntegrationFee}{28.9}
\newcommand{\rtIntegrationFeeHi}{34.7}
\newcommand{\rtIntegrationFeeLo}{20.3}
\newcommand{\rtIntegrationJointHigh}{61.4}
\newcommand{\rtMainProfiles}{21}
\newcommand{\rtNemotronErrorHigh}{0.375}
\newcommand{\rtNemotronErrorLow}{0.125}

\newcommand{\rtThresholdConst}{0.1031}
\newcommand{\rtThresholdError}{0.4211}

\newcommand{\rtThresholdWait}{2.9670}

\newcommand{\svNewCells}{66}

\newcommand{\svNewContrasts}{18}

\newcommand{\svNewPreRows}{896}

\setcopyright{none}
\date{}

\title[Who Pays for Thinking?]{Who Pays for Thinking? Billing Hidden Reasoning Tokens in Model--Cloud Contracts}

\begin{document}

\author{Chao Feng}
\affiliation{%
  \department[0]{Communication Systems Group}
  \department[1]{Department of Informatics}
  \institution{University of Zurich}
  \city{Zurich}
  \country{Switzerland}}
\email{cfeng@ifi.uzh.ch}

\author{Burkhard Stiller}
\affiliation{%
  \department[0]{Communication Systems Group}
  \department[1]{Department of Informatics}
  \institution{University of Zurich}
  \city{Zurich}
  \country{Switzerland}}
\email{stiller@ifi.uzh.ch}

\renewcommand{\shortauthors}{Chao Feng and Burkhard Stiller}

\clubpenalty=10000
\widowpenalty=10000
\displaywidowpenalty=10000

\begin{abstract}
In token-priced language model services, users can see the prompts they submit and the answers they receive, but may also pay for internal reasoning and candidate answers that are not shown in the final response. This paper develops a three-party service model to study how charging for this intermediate computation shapes providers' decisions and social welfare. Users purchase answers from a model firm that pays an execution cloud, with either supplier controlling the computation budget. The model separates user billing from upstream payments and identifies when metered pricing encourages excessive generation, when fixed service prices restrict valuable computation, and how responsibility for execution costs changes these choices. A LiveCodeBench study measures answer quality, latency, token use, and resource consumption across five model deployments, with numerical billing comparisons for the three covered by archived tariffs. At equal average charges on the reference workload, the calibrated comparisons show that billing structure changes both service participation and the computation budgets suppliers select. For accepted services with different budgets, the measured outcomes define the user-value thresholds at which their welfare ranking reverses. Metered pricing yields higher welfare when the additional inference it induces creates net social value, whereas fixed service pricing is preferable when the costs of that computation outweigh its benefits.
\end{abstract}

\ccsdesc[500]{Theory of computation~Algorithmic game theory and mechanism design}
\ccsdesc[300]{Applied computing~Economics}
\keywords{reasoning tokens, metered pricing, fixed service pricing, revenue sharing, inference economics, cloud contracts}

\maketitle
\hypersetup{pdfauthor={Chao Feng and Burkhard Stiller}}

\section{Introduction}
\label{sec:introduction}

Artificial intelligence (AI) applications built on large language models (LLMs), including chatbots and coding assistants, have become an important part of the Web. Similarweb estimates that generative AI platforms received an average of 9.5 billion Web visits per month between June 2025 and May 2026 \cite{similarweb2026genai}. Token-priced AI services can charge for internal reasoning and for candidates that an application generates but does not deliver to its end user \cite{google2026thinking,google2026candidates}. How does charging for this unseen generation shape providers' computation choices, and when does the additional computation create enough value for users to justify its cost?

In cloud-hosted services, users pay a model firm that purchases computation from a cloud provider. Existing studies analyze incentives to expand inference, price menus for users with different valuations, and manipulation of reported token counts \cite{artola2026games,bergemann2026menu,artola2026overcharging}. These studies focus on the user--provider transaction, leaving the model--cloud allocation of incentives largely implicit. A larger bill can reward the model firm while additional execution costs fall on the cloud, making both the payment between them and control over generation relevant to service provision.

This paper compares metered billing for intermediate generation with fixed service pricing at equal average charges on the reference workload. A three-party model separates this retail choice from the payment to the execution cloud and the right to choose the attempt budget. Payments for compute time, fixed capacity fees, and revenue sharing give the suppliers different returns to another candidate. The analysis determines which services satisfy participation and selects the controller's most profitable budget, retaining ties and exit. It identifies when metered pricing supports more computation and when that computation improves welfare. Pooling, which submits multiple tasks in one generation call, enters through two channels. A lower cost increment can change the cloud's preferred budget, and lower cost levels can make a previously rejected budget tradable. Either response can lower welfare even when each fixed service becomes cheaper.

A study on 106 LiveCodeBench tasks \cite{jain2024livecodebench} measures correct and incorrect delivery, latency, token generation, and resource use across five model deployments. These measurements, together with archived tariffs for three deployments, provide the inputs for the numerical pricing comparison. For these three deployments, metered bills grow faster than execution costs in every measured configuration as the budget increases. At equal average charges on the reference workload, the calibrated comparisons show how these incentives and supplier participation produce different budgets. The welfare analysis identifies when the additional inference induced by metered billing creates enough value to justify its costs and when a fixed service price produces a better outcome. These boundaries depend on the model's quality gains, execution costs, and the losses users associate with errors and waiting. Supplier profits explain the budget changes, and user purchase thresholds identify where those predicted services can be traded.

The contribution is a calibrated comparison of retail pricing under joint and separate supplier participation. Separating the model firm and cloud can leave a jointly profitable service unprovided when the upstream payment assigns a loss to one party. With a fixed capacity fee, the low-budget bill can be too small for the model firm while high-budget cost is too large for the cloud. The integrated benchmark exposes this difference, while measured service profiles determine the available budgets and the user-value conditions governing their welfare.

\section{Problem Formulation}
\label{sec:model}

Consider a model firm serving a fixed cohort of requests using resources purchased from an execution cloud. The comparison holds the trained model, execution configuration, and answer-selection rule fixed.

\subsection{Service profiles}

The model firm or cloud selects an attempt budget $k$ from a finite set $\Kset$. Each budget generates candidate answers before the selection rule returns an answer or an abstention. Its service profile is
\begin{equation}
 x_k=(q_k^+,q_k^-,\tau_k,H_k,t_k,C_k),
 \label{eq:service-profile}
\end{equation}
where $q_k^+$ and $q_k^-$ are the probabilities of correct and incorrect delivery, $\tau_k$ is latency, and $H_k$ is the number of billable intermediate tokens. The quantities $t_k$ and $C_k$ denote billed graphics processing unit (GPU) hours and real resource cost, respectively. All quantities are expected amounts per request, with shared resource use allocated over the fixed cohort. Abstention has probability $1-q_k^+-q_k^-$.

\subsection{User-pricing rules}

Metered pricing $A$ charges for intermediate generation, whereas fixed service pricing $D$ includes it in a predetermined price,
\begin{equation}
 B_A(k)=b+p_hH_k,\qquad B_D(k)=P.
 \label{eq:simplebill}
\end{equation}
Here $b$ is a common base charge, $p_h$ is the intermediate-token rate, and $P$ is independent of the selected budget. Both rules charge for an executed request, including one that ends in abstention. Their price levels are aligned using the reference-workload rule defined in Section~\ref{sec:contract-setting}.

For correct-answer value $v$, error loss $d$, waiting cost $\ell$, and outside-option value $u_0$, the user's expected surplus under pricing rule $j\in\{A,D\}$ is
\begin{equation}
 U_j(k)=vq_k^+-dq_k^--\ell\tau_k-B_j(k)-u_0.
 \label{eq:user}
\end{equation}

\subsection{Welfare comparison}

Social welfare is the value of the service after deducting its real costs,
\begin{equation}
 W(k)=vq_k^+-dq_k^--\ell\tau_k-C_k-u_0.
 \label{eq:welfare}
\end{equation}
The option $\varnothing$ denotes not undertaking the service and has zero incremental welfare. The socially efficient outcomes satisfy
\begin{equation}
 \Kset^\star=\argmax_{k\in\Kset\cup\{\varnothing\}}W(k).
 \label{eq:efficient-budgets}
\end{equation}
For a given upstream contract and control assignment, let $\mathcal O_j$ denote the service outcomes of pure subgame-perfect equilibria under pricing rule $j$, including rejection and tied choices. Section~\ref{sec:incentives} derives these outcomes from the parties' incentives and participation conditions. The welfare comparison is
\begin{equation}
 \mathcal G=\bigl\{W(k_A)-W(k_D)\,\bigm|\,
 k_A\in\mathcal O_A,\ k_D\in\mathcal O_D\bigr\}.
 \label{eq:retail-welfare-comparison}
\end{equation}
The problem is to identify the conditions under which metered or fixed pricing induces higher welfare. Metered pricing is strictly better across all these outcomes when $\min\mathcal G>0$. Fixed pricing is strictly better when $\max\mathcal G<0$. Other cases retain the dependence on equilibrium selection.

\section{Method}
\label{sec:incentives}

The analysis derives the conditions under which metered billing or fixed service pricing yields higher social welfare. It solves for service participation and the controller's budget choice while holding service profiles, upstream contract terms, and control rights fixed.

\begin{figure*}[t]
 \centering
 \includegraphics[width=0.82\textwidth]{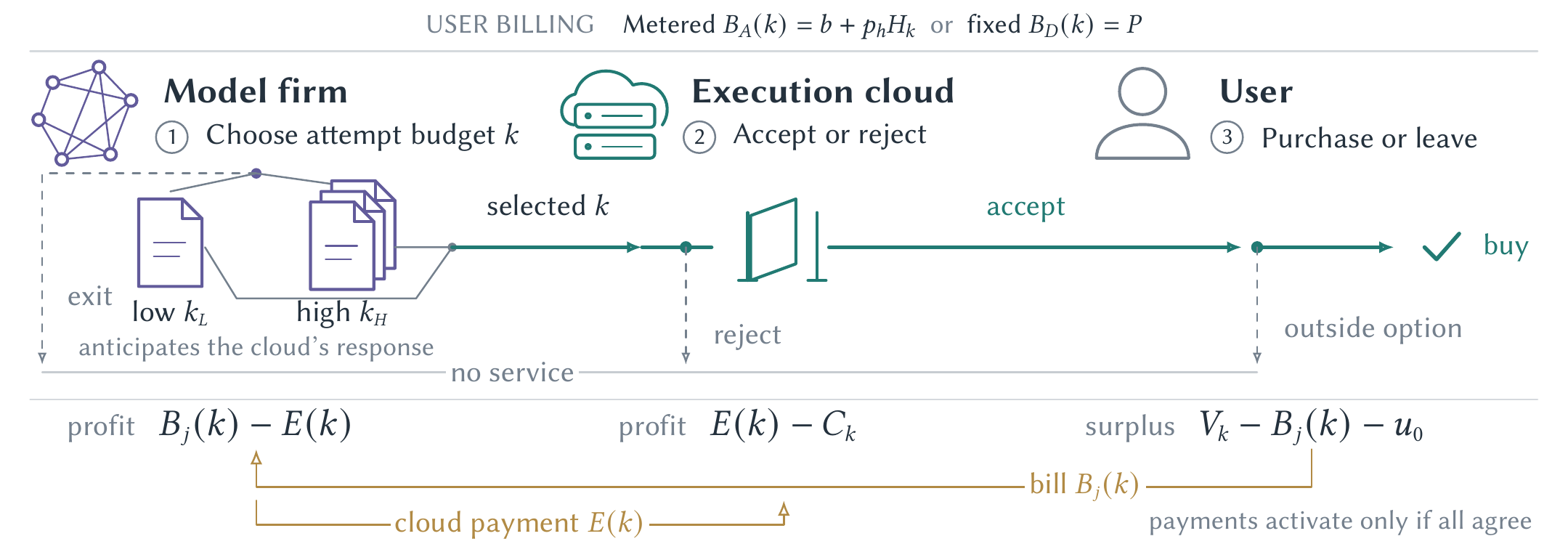}
 \caption{Service decisions and payments.}
 \Description{Under metered or fixed user pricing, the model firm chooses a low or high attempt budget or exits. The cloud accepts or rejects. The user purchases or takes an outside option. The trade path is shown in teal and no-trade exits in dashed gray. Conditional on trade, the user pays B sub j of k to the model firm and the model firm pays E of k to the cloud. The model firm's profit is the user bill minus the upstream payment. The cloud's profit is the upstream payment minus resource cost. User surplus is service value minus the bill and the outside-option value. The diagram shows model control. Under cloud control, the cloud chooses the budget and the model firm accepts or rejects.}
 \label{fig:service-game}
\end{figure*}

\subsection{Contract setting}
\label{sec:contract-setting}

The two-budget results use $\Kset=\{k_L,k_H\}$ with $k_H>k_L$, holding the model's native reasoning mode fixed. Training cost is sunk. For a fixed user-pricing rule $j$, write $B_k=B_j(k)$ and $U_k=U_j(k)$. A reference distribution $w_k$, fixed before service choice, aligns the two prices through
\begin{equation}
 P=\sum_{k\in\Kset}w_k B_A(k),\qquad \sum_{k\in\Kset}w_k=1.
 \label{eq:price-reference}
\end{equation}
The expected value before payment is
\begin{equation}
 V_k=vq_k^+-dq_k^--\ell\tau_k.
 \label{eq:value}
\end{equation}
Risk-neutral users evaluate the publicly committed distribution of quality, bills, and latency. Their outside option can include local inference, net of its resource costs.

User pricing shapes the model firm's revenue, while the upstream payment determines who benefits from additional computation. Under GPU-hour purchasing, the firm pays $E(k)=r_g t_k$ at a contracted rate $r_g>0$. Its token revenue and cloud payment can grow at different rates. Microsoft documents OpenAI's Azure commitments and capped revenue-sharing payments \cite{microsoft2025commitment,microsoft2026share}. Revenue sharing gives the cloud a direct stake in the service bill. The firm transfers a fraction $\alpha$ of the user's payment,
\begin{equation}
 E(k)=\alpha B_k,\qquad 0<\alpha<1.
 \label{eq:transfer}
\end{equation}
In the baseline, the cloud bears execution cost and the firm's additional cost is zero. Metered generation raises both suppliers' revenues under sharing. Fixed service pricing holds both revenues constant as the cloud's execution cost rises. A fixed capacity fee provides the alternative $E(k)=L$, where $L$ is the per-request share of a fee offered for the cohort, payable on acceptance.

Control $\iota\in\{M,C\}$ belongs to the firm or cloud independently of payment. Cloud control represents managed inference with delegated scheduling and retries, while model control represents a firm using its own scheduler. These arrangements motivate treating control as an institutional variable.

Pooling permission $\pi\in\{0,1\}$ changes the feasible implementation menu. The cost-only pooling comparison holds quality, bills, and latency fixed at each budget. Its minimum resource cost is $C_\pi(k)$, with $C_1(k)\leq C_0(k)$ because the old configuration remains available. Appendix~\ref{sec:identification} treats implementation-dependent bills and quality.

The game models an offer to serve a future request cohort. Terms, pooling permission, and the menu are fixed before the controller chooses a budget or exit. The other supplier accepts or rejects, and the user purchases or takes their outside option. Payment and capacity activation follow acceptance by all parties. Figure~\ref{fig:service-game} shows model control, with the first two roles exchanged under cloud control. Rejection and exit yield zero incremental payoffs.

Amazon EC2's one- and three-year commitments illustrate slow price adjustment \cite{aws2026reserved}. The activation sequence and fixed cohort share define the modeled service commitment. Appendix~\ref{app:contract-benchmarks} discusses flexible pricing and contract design.

For an activated service, supplier profits are
\begin{align}
 \Pi_M(k)&=B_k-E(k),\label{eq:profits}\\
 \Pi_C(k)&=E(k)-C_\pi(k).\nonumber
\end{align}
Adding user surplus recovers equation~\eqref{eq:welfare} with $C_k=C_\pi(k)$. Each result states whether strict participation is required at both budgets or only at realized trades. Zero-profit and zero-surplus ties retain acceptance and rejection.

Under strict user purchase, the integrated benchmark maximizes $B_j(k)-C_\pi(k)$ over budgets and exit. Separation additionally requires nonnegative profit for each supplier.

\subsection{How pricing changes computation and welfare}
\label{sec:pricing-analysis}

In Figure~\ref{fig:service-game}, rejection branches determine the available services, and supplier payoffs rank the accepted budgets. Summing the three payoffs gives welfare. The following results connect these decisions to user pricing.

Consider two budgets $k_H>k_L$, and write $\Delta X=X_H-X_L$. The baseline assigns execution cost to the cloud, with no additional model-firm cost. For a fixed upstream contract, $E_j(k)$ denotes the payment under user-pricing rule $j$.

\begin{proposition}[Budget choice]\label{prop:retail-budget}
Suppose both budgets give the user and both suppliers strictly positive payoffs under each pricing rule. The controller's high-minus-low profit is given in Table~\ref{tab:retail-margins}. A positive difference selects $k_H$, a negative difference selects $k_L$, and zero retains both.
\end{proposition}
\begin{proof}
Subtracting profits across budgets gives
\begin{equation}
 \Delta\Pi_M^j=\Delta B_j-\Delta E_j,\qquad
 \Delta\Pi_C^j=\Delta E_j-\Delta C.
 \label{eq:retail-margins}
\end{equation}
Substituting each payment rule and using $\Delta B_A=p_h\Delta H$ and $\Delta B_D=0$ gives the table. Strict participation makes both budgets available and excludes exit.
\end{proof}

\begin{table}[t]
 \caption{Private returns to a larger budget.}
 \label{tab:retail-margins}
 \centering
 \begin{tabular}{@{}llcc@{}}
 \toprule
 Payment & Control & Metered $A$ & Fixed $D$ \\
 \midrule
 GPU hours & $M$ & $p_h\Delta H-r_g\Delta t$ & $-r_g\Delta t$ \\
 GPU hours & $C$ & $r_g\Delta t-\Delta C$ & $r_g\Delta t-\Delta C$ \\
 Capacity fee & $M$ & $p_h\Delta H$ & $0$ \\
 Capacity fee & $C$ & $-\Delta C$ & $-\Delta C$ \\
 Revenue share & $M$ & $(1-\alpha)p_h\Delta H$ & $0$ \\
 Revenue share & $C$ & $\alpha p_h\Delta H-\Delta C$ & $-\Delta C$ \\
 \bottomrule
 \end{tabular}
\end{table}

For positive $\Delta H$, $\Delta t$, and $\Delta C$, fixed pricing strictly favors the smaller budget under model control with GPU-hour purchasing and under cloud control with either a capacity fee or sharing. A cloud paid a capacity fee makes the same choice under metered billing. With zero model-firm cost, fixed pricing leaves the firm indifferent under capacity fees and sharing.

Participation can change the outcome even when the controller's ranking remains unchanged. Revenue sharing makes this connection explicit because the cloud accepts a budget only when its share covers execution cost.

\begin{proposition}[A participation-induced switch]\label{prop:retail-participation}
Suppose the user strictly prefers purchasing both budgets under either pricing rule. Under revenue sharing, if
\begin{equation}
 \alpha B_A(k_L)<C_L<\alpha P<C_H<\alpha B_A(k_H),
 \label{eq:retail-participation-switch}
\end{equation}
metered billing uniquely selects $k_H$, while fixed pricing uniquely selects $k_L$, under either assignment of control.
\end{proposition}
\begin{proof}
Under metered billing, the cloud loses money at $k_L$ and earns positive profit at $k_H$. Under fixed pricing, it earns positive profit at $k_L$ and loses money at $k_H$. The model firm retains a positive share of each bill. Each pricing rule therefore supports exactly one profitable service, which is preferred to exit by either controller.
\end{proof}

Fixed pricing redistributes revenue toward the smaller budget, changing participation. Proposition~\ref{prop:retail-budget} covers the complementary case where both budgets remain available and pricing changes their ranking.

\begin{proposition}[Welfare ranking]\label{prop:retail-welfare}
If metered billing selects $k_H$, fixed pricing selects $k_L$, and both services are purchased, their welfare difference is
\begin{align}
 W_A-W_D&=\Delta V-\Delta C\nonumber\\
 &=v\Delta q^+-d\Delta q^--\ell\Delta\tau-\Delta C.
 \label{eq:retail-welfare-gap}
\end{align}
Metered billing yields higher welfare when $\Delta V>\Delta C$. Fixed pricing yields higher welfare when $\Delta V<\Delta C$.
\end{proposition}
\begin{proof}
Summing user surplus and supplier profits cancels the retail bill and upstream payment. Subtracting welfare across the selected budgets gives the result.
\end{proof}

For $\Delta q^+>0$, the welfare ranking reverses at
\begin{equation}
 v^\star=\frac{d\Delta q^-+\ell\Delta\tau+\Delta C}{\Delta q^+}.
 \label{eq:retail-value-threshold}
\end{equation}
Above this threshold, the larger budget creates net social value, and fixed pricing induces insufficient computation. Below it, metered billing induces excessive computation. The empirical analysis estimates the quality, latency, and resource differences that determine this threshold and evaluates user participation at the same values.

When a pricing rule produces no service, its welfare is zero. At participation or budget ties, the analysis retains every pure equilibrium outcome and compares the resulting welfare-difference set $\mathcal G$.

\subsection{How pooling changes the pricing comparison}
\label{sec:pooling}

Pooling can reduce the cost of each budget while changing which budget suppliers choose. Consider revenue sharing at a fixed rate $\alpha$, and let $\pi=0,1$ denote the original and pooled configurations. To isolate the cost effect, hold quality, latency, and retail bills fixed at each budget. Define
\begin{equation}
 s_k=C_0(k)-C_1(k)\geq0,\qquad
 \Delta C_\pi=C_\pi(k_H)-C_\pi(k_L).
 \label{eq:retail-pooling-costs}
\end{equation}
The change in the cost increment satisfies $\Delta C_0-\Delta C_1=s_H-s_L$. Pooling makes the larger budget relatively cheaper when it saves more at $k_H$ than at $k_L$.

If the selected service changes from the original low budget to the pooled high budget, welfare changes by
\begin{equation}
 W_1(k_H)-W_0(k_L)=\Delta V-\Delta C_1+s_L.
 \label{eq:retail-pooling-welfare}
\end{equation}
The gain combines the value of additional computation with the saving at the original budget, less the pooled cost of increasing the budget.

Pooling affects two different decisions in Figure~\ref{fig:service-game}. Under cloud control, lower incremental cost can reverse the cloud's ranking of the two budgets. Under model control, lower cost can turn the cloud's rejection of the proposed high budget into acceptance.

\begin{proposition}[A change in budget incentives]\label{prop:retail-pooling-margin}
Suppose the cloud controls the budget and both budgets give all participants strictly positive payoffs under both pricing rules and configurations. If
\begin{equation}
 0<\Delta C_1,\qquad
 \Delta V+s_L<\Delta C_1<\alpha p_h\Delta H<\Delta C_0,
 \label{eq:retail-pooling-margin}
\end{equation}
pooling changes the metered-pricing outcome from $k_L$ to $k_H$ and lowers welfare. Under fixed service pricing, the cloud continues to choose $k_L$, and welfare rises by $s_L$.
\end{proposition}
\begin{proof}
The metered-billing margin changes from $\alpha p_h\Delta H-\Delta C_0<0$ to $\alpha p_h\Delta H-\Delta C_1>0$. Under fixed pricing, the margin remains negative at $-\Delta C_\pi$. Equation~\eqref{eq:retail-pooling-welfare} gives the stated welfare changes.
\end{proof}

Here, pooling lowers the cost of additional billable generation enough to change the cloud's choice. Every fixed service becomes weakly cheaper, but the selected service uses more computation than its additional value justifies.

Pooling can also change the outcome by making a previously rejected budget profitable.

\begin{proposition}[A change in service participation]\label{prop:retail-pooling-admission}
Suppose the model firm controls the budget, $\Delta B_A>0$, and users strictly prefer purchasing both budgets under both pricing rules. If
\begin{align}
 C_0(k_L)&<\alpha B_A(k_L),\nonumber\\
 C_1(k_H)&<\alpha B_A(k_H)<C_0(k_H),
 \label{eq:retail-pooling-admission}
\end{align}
pooling changes the metered-pricing outcome from $k_L$ to $k_H$. Welfare falls if $\Delta V+s_L<\Delta C_1$. If additionally
\begin{equation}
 C_0(k_L)<\alpha P<C_1(k_H),
 \label{eq:retail-fixed-low}
\end{equation}
fixed pricing continues to select $k_L$ in both configurations.
\end{proposition}
\begin{proof}
Under metered billing, the model firm prefers $k_H$ because it retains $(1-\alpha)B_A(k)$. The cloud initially rejects that budget and accepts it after pooling. It accepts $k_L$ throughout. Under equation~\eqref{eq:retail-fixed-low}, the cloud accepts only $k_L$ under fixed pricing in either configuration. Equation~\eqref{eq:retail-pooling-welfare} establishes the welfare sign.
\end{proof}

This mechanism depends on cost levels. It can operate even when pooling saves the same amount at both budgets, leaving $\Delta C_\pi$ unchanged. If fixed pricing instead makes both budgets acceptable, the model firm earns $(1-\alpha)P$ at either budget, and both choices remain possible. The welfare comparison then retains that equilibrium multiplicity.

\subsection{Cost responsibility and larger budget menus}
\label{sec:cost-responsibility}

The preceding results assign execution cost to the cloud. If the model firm also bears a disjoint real cost $C_M(k)\geq0$, its profit and total welfare become
\begin{align}
 \Pi_M^j(k)&=B_j(k)-E_j(k)-C_M(k),\nonumber\\
 W(k)&=V_k-C_M(k)-C(k)-u_0.
 \label{eq:shared-cost-welfare}
\end{align}
Every model-firm margin in Table~\ref{tab:retail-margins} decreases by $\Delta C_M$, while the social cost increment becomes $\Delta C_M+\Delta C$.

This cost allocation changes the effect of fixed pricing. Under model control and revenue sharing, the two profit increments are
\begin{align}
 \Delta\Pi_M^A&=(1-\alpha)p_h\Delta H-\Delta C_M,\nonumber\\
 \Delta\Pi_M^D&=-\Delta C_M.
 \label{eq:retail-model-cost}
\end{align}
If $0<\Delta C_M<(1-\alpha)p_h\Delta H$ and both budgets strictly satisfy participation under both pricing rules, metered pricing selects $k_H$ while fixed pricing selects $k_L$. Bearing part of the additional cost turns the model firm's fixed-price indifference into a strict preference for less computation. The strict pooling results also persist for sufficiently small positive reallocations of a fixed total cost to the model firm, as shown in Appendix~\ref{app:generalizations}.

The same distinction between incentives and participation extends to larger budget menus. Fix a retail rule, a finite ordered menu $\Kset$, and the bill and upstream payment schedules across configurations. Assume $C_1(k)\leq C_0(k)$ for every $k$, and write $G_\pi(k)=B_k-E(k)-C_{M,\pi}(k)$ for the model firm's profit.

\begin{proposition}[Finite budget menus]\label{prop:finite-budgets}
Under cloud control, suppose users and the model firm strictly accept every budget in both configurations, and the cloud's best initial profit is positive. If $s_k=C_0(k)-C_1(k)$ is nondecreasing in $k$, the smallest and largest elements of
\[
 K_\pi=\argmax_{k\in\Kset}\{E(k)-C_\pi(k)\}
\]
weakly increase after pooling.

Under model control, define $A_\pi=\{k\in\Kset:E(k)>C_\pi(k)\}$. Suppose cloud participation has no zero-profit ties, $A_0\ne\varnothing$, users strictly accept every budget, and $G_\pi$ is positive and strictly increasing on each $A_\pi$. Then the selected budget is $k_\pi=\max A_\pi$, with $k_1\geq k_0$.
\end{proposition}
\begin{proof}
For cloud optima $x\in K_0$ and $y\in K_1$, let $f_\pi=E-C_\pi$. Optimality gives $0\leq f_0(x)-f_0(y)\leq s_y-s_x$. Nondecreasing savings order the extreme optima. Under model control, lower costs imply $A_0\subseteq A_1$, and the firm chooses each set's largest element. Appendix~\ref{app:generalizations} gives the full argument.
\end{proof}

The cloud-control result applies increasing differences \cite{topkis1978,milgrom1994}. Appendix~\ref{app:generalizations} shows that its savings condition is tight for monotonicity uniformly over baseline profits.

The revenue share also affects these choices. On a common strictly participating menu with nondecreasing $B_k$, increasing $\alpha$ weakly raises the extreme cloud optima because $\partial[\alpha B_k-C(k)]/\partial\alpha=B_k$. For nonnegative bills, it also expands the cloud's acceptance set. When the model firm bears its own cost, its participation condition must be checked alongside this expansion.

For any realized transition between traded budgets $k_0$ and $k_1$, write total cost as $T_\pi=C_{M,\pi}+C_\pi$. The welfare change is
\begin{align}
 W_1(k_1)-W_0(k_0)
 &=V_{k_1}-V_{k_0}\nonumber\\
 &\quad-[T_1(k_1)-T_1(k_0)]\nonumber\\
 &\quad+T_0(k_0)-T_1(k_0).
 \label{eq:menu-welfare}
\end{align}
A budget increase is beneficial when its additional value and the saving at the original budget exceed its additional cost. The experiments compare pricing rules within a fixed configuration using equations~\eqref{eq:retail-welfare-gap}--\eqref{eq:retail-value-threshold}.

\section{Experiments}
\label{sec:measurement}

The evaluation combines measured service profiles with the contract model to compare metered billing and fixed pricing at equal average charges on the reference workload. It determines participation and budget choice, then identifies the user-value thresholds that reverse their welfare ranking. Sensitivity covers deployments, execution configurations, upstream terms, and resource costs.

\subsection{Experimental setup and calibration}
\label{sec:retail-setup}

The study uses 106 evaluation tasks from LiveCodeBench \cite{jain2024livecodebench}, with 16 separate tasks reserved for calibration. Five deployments cover Qwen3.8 with thinking enabled and disabled, Qwen3.5 thinking, Nemotron thinking, and Coder. Each deployment is evaluated with one and eight candidate attempts in three \emph{repetitions}, denoted $r_0$, $r_1$, and $r_2$. All three use the same 106 tasks, with a different set of generation seeds in each repetition. Decoding settings and per-candidate generation limits remain fixed within each deployment.\footnote{Code and raw experimental data at \url{https://github.com/luke-feng/whopayforthinking}.}

Serial execution processes one task per call, with its candidate streams decoded together. Batch16 and batch64 process up to 16 and 64 tasks per call, respectively. Every deployment includes serial and batch16 execution, while Qwen3.8 non-thinking additionally includes batch64. Within each repetition, the compared budget and pooling settings run in the same GPU allocation. The complete service study contains \svNewCells{} measured cells. A priced profile combines one deployment, configuration, and repetition. The three quoted deployments give seven deployment/configuration combinations and \rtMainProfiles{} profiles.

The service delivers the first candidate passing fixed screening tests and abstains if none passes. An offline audit distinguishes correct from incorrect delivery. The measured window runs from cohort dispatch through the final answer or abstention, including online selection. Figure~\ref{fig:delivery-composition} shows batch16 outcomes averaged over three repetitions. The on/off labels denote Qwen3.8 thinking modes.

\begin{figure}[t]
 \centering
 \includegraphics[width=\columnwidth]{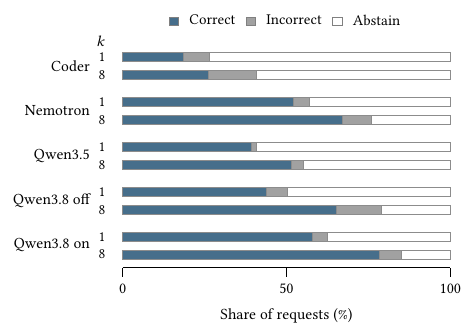}
 \caption{Delivery outcomes under batch16.}
 \Description{Horizontal stacked bars show correct delivery, incorrect delivery, and abstention for each of five deployments. Each deployment has an upper bar for one candidate and a lower bar for eight candidates. Values are descriptive means of three repetitions. Segments sum to 100 percent.}
 \label{fig:delivery-composition}
\end{figure}

Archived tariff anchors support numerical billing comparisons for Qwen3.5 and the two Qwen3.8 deployments. Within each deployment and configuration, the common base charge $b$ is the mean input and delivered-output charge across both budgets and all three repetitions. Metered billing adds the charges for reasoning and unselected candidate tokens. The fixed service price is
\begin{equation}
 P(w)=(1-w)\overline B_A(1)+w\overline B_A(8),
 \label{eq:retail-calibration}
\end{equation}
where $\overline B_A(k)$ averages the metered bill across repetitions. The reference setting uses $w=1/2$. Each resulting price remains fixed across repetitions and budget choices, including responses ending in abstention. The tariff sources and their deployment correspondence are recorded in Appendix~\ref{app:tariffs}.

The numerical contract comparisons use a reference cloud share $\alpha=1/2$. A fixed-capacity comparison uses $L=P(1/2)/2$, matching the reference average upstream payment. Both control assignments are evaluated with $C_M=0$, retaining exit and all ties. Resource costs combine GPU occupancy, central processing unit (CPU) use, and estimated energy consumption. The capital-and-hosting rate is $\rho=2.4$ USD per busy GPU hour, with sensitivity over 1.9--3.4 USD. Energy is estimated from power ratings and declared power-draw ranges. Appendix~\ref{app:anchors} gives the cost assumptions, and a separate cost scope includes startup.

This counterfactual analysis follows the service measurements and preserves each repetition's hardware identity. Appendix~\ref{app:timeline} distinguishes it from the registered analyses. Common profiles and reference-matched charges isolate how payment incentives and participation change the supplied service.

\subsection{Service participation and budget choice}
\label{sec:retail-choices}

Metered bills grow faster than service costs across all \rtMainProfiles{} priced profiles at the central cost scenario, under both full-token billing and the common-base construction. For serial execution, reconstructed bills rise \prSerialBillLo--\prSerialBillHi{} times from one to eight candidates, while costs rise \prSerialCostLo--\prSerialCostHi{} times. Under sharing, the larger metered service requires a smaller break-even cloud share in every profile, $C_H/B_A(k_H)<C_L/B_A(k_L)$. This ordering is independent of the chosen share. Fixed pricing reverses the ordering because its equal revenue must cover a higher cost at the larger budget.

Under sharing and strict user purchase of both budgets, metered billing selects eight in both Qwen3.8 deployments. Fixed pricing selects one except for a model-controlled thinking tie in $r_2$. Qwen3.5 supplies neither service. Appendix~\ref{app:supplier-participation} lists the outcomes for each repetition.

The non-thinking deployment illustrates Proposition~\ref{prop:retail-participation}. Metered billing leaves the cloud with a loss at one candidate and a profit at eight. Fixed pricing reverses these signs, supporting one profitable budget under either controller.

\begin{figure}[t]
 \centering
 \includegraphics[width=\columnwidth]{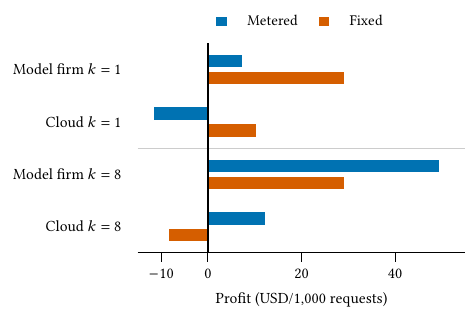}
 \caption{Qwen3.8 off profits ($r_0$, batch16, $\alpha=0.5$).}
 \Description{Paired blue metered and orange fixed-price bars compare model-firm and cloud profits for one and eight candidates. The case is Qwen3.8 non-thinking, batch16, repetition r0, with reference-matched prices and a one-half cloud revenue share. Negative cloud profit occurs at one candidate under metered billing and eight under fixed pricing. These are proposed-service payoffs conditional on execution, not realized losses on rejected services.}
 \label{fig:supplier-profits}
\end{figure}

Figure~\ref{fig:supplier-profits} gives the $r_0$ profits conditional on execution, with bills and costs in Table~\ref{tab:retail-payoffs}. Equal average charges on the reference workload lead to different realized bills.

The thinking deployment also exposes the role of control. In repetition $r_2$, fixed pricing supports both budgets and gives the model firm the same profit at each. Model control retains $\{1,8\}$. The cloud receives the same payment at both budgets and selects the cheaper one. User purchase thresholds determine where these supplier choices result in trade.

Table~\ref{tab:supplier-structure} compares integration and separation for Qwen3.8 non-thinking, batch16. Its outcomes hold in all three repetitions, under strict user purchase and either controller for the separated suppliers. Integration and sharing support the same budgets. Under the fixed capacity fee, metered pricing instead provides no service. The low-budget bill is below the firm's fee, while high-budget cost exceeds the cloud's payment. In $r_0$, that high-budget service would yield \rtIntegrationJointHigh{} USD of joint profit per 1,000 requests, yet the cloud would lose \rtIntegrationCloudLoss{} USD. Fixed pricing restores trade at one candidate. Separate participation determines whether positive joint profit can be realized.

\begin{table}[t]
\caption{Supplier structure (Qwen3.8 off, batch16).}
\label{tab:supplier-structure}
\centering\small
\begin{tabular}{@{}lcc@{}}
\toprule
Supply structure & Metered $A$ & Fixed $D$ \\
\midrule
Integrated & $8$ & $1$ \\
Revenue sharing & $8$ & $1$ \\
Capacity fee & $\varnothing$ & $1$ \\
\bottomrule
\end{tabular}
\end{table}

\subsection{Welfare comparison}
\label{sec:retail-welfare-results}

The budget differences in Section~\ref{sec:retail-choices} determine the services whose welfare must be compared. For Qwen3.8 non-thinking, batch16, $r_0$, increasing the budget from one to eight adds \rtAdditionalCorrect{} correct deliveries and \rtAdditionalWrong{} incorrect deliveries across the 106 tasks. It also increases latency and resource use. These increments give the approximate welfare threshold
\begin{equation}
 v^\star_{r_0}\approx\rtThresholdConst+\rtThresholdError d+\rtThresholdWait h,
 \label{eq:retail-case-value}
\end{equation}
where $v$ is correct-answer value, $d$ is error loss, and $h$ is the cost of an hour of waiting. Here $h=3600\ell$ when latency is recorded in seconds. Using the unrounded threshold, when both selected services are purchased,
\begin{equation}
 W_A-W_D=\frac{\rtAdditionalCorrect}{106}(v-v^\star_{r_0}).
 \label{eq:retail-case-welfare}
\end{equation}
Metered billing yields higher welfare above this threshold, and fixed pricing yields higher welfare below it.

User participation determines where this comparison applies. The corresponding purchase thresholds, shown with rounded coefficients, are
\begin{align}
 \underline v_A&\approx\rtBuyMeteredConst+\rtBuyMeteredError d+\rtBuyMeteredWait h+\rtBuyMeteredOutside u_0,\nonumber\\
 \underline v_D&\approx\rtBuyFixedConst+\rtBuyFixedError d+\rtBuyFixedWait h+\rtBuyFixedOutside u_0.
 \label{eq:retail-case-purchase}
\end{align}
A fixed-price welfare advantage requires $\max\{\underline v_A,\underline v_D\}<v<v^\star_{r_0}$. Metered billing is preferable when $v$ exceeds both purchase thresholds and the welfare threshold. In the capacity-fee case of Table~\ref{tab:supplier-structure}, fixed pricing yields positive welfare whenever its low-budget service is purchased, while metered pricing yields zero because neither budget trades.

\begin{figure}[t]
 \centering
 \includegraphics[width=\columnwidth]{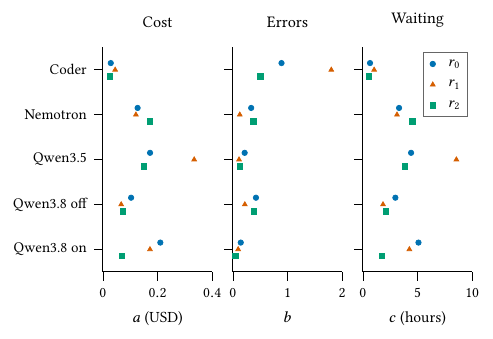}
 \caption{Threshold coefficients (batch16, $\rho=2.4$).}
 \Description{Three dot-plot panels show the cost, error-loss, and waiting-time coefficients in the welfare threshold v star equals a plus b times d plus c times h. Rows show five deployments under batch16 at the central resource-cost scenario. Blue circles, orange triangles, and green squares identify repetitions r0, r1, and r2. The legend is at the upper right of the Waiting panel. The horizontal scales have different units.}
 \label{fig:threshold-coefficients}
\end{figure}

Figure~\ref{fig:threshold-coefficients} reports $(a,b,c)$ in $v^\star=a+bd+ch$ by deployment and repetition, under batch16 at the central cost scenario. Its panels use separate scales. In the tied thinking case, selecting eight under both rules gives equal welfare, while selecting one under fixed pricing gives the comparison in Proposition~\ref{prop:retail-welfare}.

Coder and Nemotron contribute complete service profiles and social budget thresholds. Coder has a low additional resource cost, but its ratio of additional incorrect to correct deliveries ranges from \rtCoderErrorLow{} to \rtCoderErrorHigh{}, compared with \rtNemotronErrorLow{} to \rtNemotronErrorHigh{} for Nemotron. Error losses play a larger role in Coder's budget threshold. For Qwen3.5, $r_1$ has the largest cost and waiting-time coefficients among its three repetitions. For Coder and Nemotron, the required retail-price thresholds are reported in Appendix~\ref{app:retail-details}. For the services with different selected budgets, the welfare ranking turns on whether the additional correct answers justify the accompanying errors, waiting time, and resource costs.

\subsection{Sensitivity analysis}
\label{sec:retail-sensitivity}

\begin{figure}[t]
 \centering
 \includegraphics[width=\columnwidth]{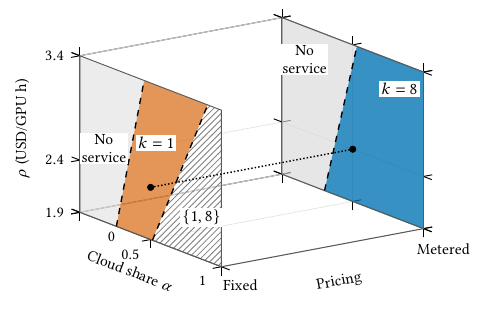}
 \caption{Budget regions under two pricing rules.}
 \Description{The front slice represents fixed service pricing and the rear slice metered billing. The horizontal axis is the cloud revenue share and the vertical axis the capacity-cost rate from 1.9 to 3.4 USD per GPU hour. Under metered billing there are no-service and eight-attempt regions. Under fixed pricing there are no-service, one-attempt, and tied one-or-eight regions. Connected reference points have the same cost and share but select one under fixed pricing and eight under metered billing. Dashed participation boundaries retain all supported tied outcomes.}
 \label{fig:pricing-phase}
\end{figure}

The Qwen3.8 non-thinking budget difference persists beyond the reference price and share. Across all three batch16 repetitions, both service and startup-inclusive cost accounting, and the evaluated energy endpoints, metered billing selects eight candidates while fixed pricing selects one within the strict switching intervals summarized in Table~\ref{tab:retail-sensitivity}.

\begin{table}[t]
\caption{Approximate budget-switch ranges.}
\label{tab:retail-sensitivity}
\centering\small
\begin{tabular}{@{}lcl@{}}
\toprule
Parameter & Approx. range & Fixed parameters \\
\midrule
Reference weight $w$ & 0.35--0.63 & $\rho=2.4,\ \alpha=0.5$ \\
Capacity cost $\rho$ & 2.00--3.08 & $w=0.5,\ \alpha=0.5$ \\
Cloud share $\alpha$ & 0.39--0.60 & $w=0.5,\ \rho=2.4$ \\
\bottomrule
\end{tabular}
\end{table}

Each row varies one parameter while retaining the stated reference values for the others. Displayed endpoints are rounded. At the exact boundaries, participation or budget ties remain possible, as recorded in the reproducibility materials. Lower fixed prices can prevent service provision, while higher prices can make both budgets acceptable and leave the model firm indifferent between them.

Figure~\ref{fig:pricing-phase} shows the joint cost--share comparison for Qwen3.8 non-thinking, batch16, $r_0$. It conditions on strict user purchase of both budgets and uses model control. The front and rear slices retain the same service profiles and price calibration. At the marked reference point $(\alpha,\rho)=(0.5,2.4)$, fixed pricing selects one candidate and metered billing selects eight. Hatching marks the larger-share region where both budgets are optimal under fixed pricing. Dashed boundaries retain the participation ties.

Reconstructing the complete token bill instead of using the common base charge preserves all \rtMainProfiles{} main-tariff budget profiles. An alternative promotional tariff changes participation in one profile near the profitability boundary. Input and delivered-output charges matter at that boundary.

Pooling reduces the additional GPU time from one to eight candidates in all \svNewContrasts{} within-repetition comparisons, removing \prGpuShareLo--\prGpuShareHi\% of the serial increment. Service accounting includes online answer selection, so these reductions cover the provision window.

Across all \svNewPreRows{} service scenarios under complete non-delivered-token billing, a share covering the small serial service also covers the large one. This leaves the pooling participation band in Proposition~\ref{prop:retail-pooling-admission} empty (Appendix~\ref{app:pooling-admission}). The retail comparisons hold execution fixed and change user pricing.

Model performance changes the welfare side. In the two Qwen3.8 comparisons using matching GPU models across thinking modes, thinking improves correct delivery and reduces errors at a given budget, while increasing cost and waiting time. Higher error losses favor that quality improvement, whereas higher waiting costs raise the value required to justify it. These differences change the welfare thresholds even when the pricing rules induce the same budget pattern.

Across the evaluated price weights, shares, and resource rates, participation and budget indifference determine the boundaries of the Qwen3.8 non-thinking budget switch. The supplemental material preserves the original pooling diagnostics, pipeline measurements, and A100 extension.

\section{Related Work}
\label{sec:related}

\begin{table}[t]
 \caption{Related pricing and contracting models.}
 \label{tab:related-comparison}
 \centering\small
 \begin{tabular}{@{}>{\raggedright\arraybackslash}p{.19\columnwidth}@{\hspace{5pt}}>{\raggedright\arraybackslash}p{.26\columnwidth}@{\hspace{5pt}}>{\raggedright\arraybackslash}p{.20\columnwidth}@{\hspace{5pt}}>{\raggedright\arraybackslash}p{.28\columnwidth}@{}}
 \toprule
 Study & Focus & Cost bearers & Evidence \\
 \midrule
 Compute games~\cite{artola2026games} & Inference budget and auction & Provider & Math/science benchmarks \\
 Menu pricing~\cite{bergemann2026menu} & Menus and model competition & Provider & Provider plans \\
 Tokenization \cite{artola2026overcharging} & Token reporting & Provider & LLM tokenization experiments \\
 Revenue sharing~\cite{cachon2005revenue} & Price and quantity & Supplier, retailer & Analytical model \\
 \textbf{This work} & Budget and participation & Model firm, cloud & Measured service profiles \\
 \bottomrule
 \end{tabular}
\end{table}

\subsection{Metered pricing and strategic inference}

Payments for services that users cannot easily evaluate create incentives to change provision or billing, a central concern of credence-goods theory \cite{dulleck2006credence}. Test-Time Compute Games studies strategic inference budgets among competing providers and proposes a price-and-quality auction \cite{artola2026games}. Menu pricing derives token-budget contracts for heterogeneous users and extends them to model competition \cite{bergemann2026menu}. Tokenization research studies profitable misreporting and mechanisms for truthful billing \cite{artola2026overcharging}. These studies connect user payments to provider behavior. Table~\ref{tab:related-comparison} locates the present comparison, which holds tariffs and token counts fixed while separating the model firm's and cloud's participation decisions.

\subsection{Payments and cost responsibility}

The allocation of marginal returns and costs is a standard concern of industrial organization \cite{tirole1988}. Revenue-sharing contracts can coordinate retail quantities and prices under specified supply-chain conditions \cite{cachon2005revenue}. Collaborative-service models relate contract choice to the buyer's and vendor's contributions to output \cite{roels2010contracting}. Cloud services add metered infrastructure and different commitment and availability terms \cite{armbrust2009clouds,aws2026reserved,aws2026spot}. Here a fixed user price removes the retail revenue increment, while a fixed capacity fee removes the model firm's additional upstream payment. The integrated comparison shows how separate cost responsibility can prevent a jointly profitable inference service from being supplied.

\subsection{Quality, selection, and service cost}

Test-time scaling studies how additional computation improves answer quality \cite{snell2024scaling}. Verification and self-consistency determine which gains reach users \cite{cobbe2021verifiers,wang2023consistency}, while imperfect verifiers can limit resampling benefits \cite{stroebl2024limits}. Quality-based and adaptive contracts use automated evaluations to determine payments \cite{saig2024statistical,saig2026adaptive}. Serving research changes execution costs through scheduling, key-value (KV) cache management \cite{yu2022orca,kwon2023pagedattention}, and the separation or interleaving of prompt processing and token generation \cite{zhong2024distserve,agrawal2024sarathi}. Analytical inference models link hardware use to efficiency \cite{pope2022efficient}. This study measures delivered correctness, errors, latency, and resources through online answer selection, then combines these service profiles with supplier participation to compare billing rules and welfare thresholds.

\section{Conclusion}
\label{sec:conclusion}

Metered billing rewards additional intermediate computation with additional revenue, while fixed service pricing removes that revenue increment. How this changes the attempt budget depends on control and cost responsibility. Under revenue sharing, metered billing rewards both suppliers for a larger bill. With fixed pricing, the cloud receives the same payment as execution cost rises, while a model firm bearing no additional cost can remain indifferent between budgets. Under a fixed capacity fee, separate participation can block a jointly profitable service that an integrated supplier would provide.

Under revenue sharing and user purchase, the calibrated comparisons show that equal average charges on the reference workload can support different computation budgets. Where metered billing induces a larger budget and both resulting services are purchased, it yields higher welfare when the value of additional correct answers exceeds the added error losses, waiting costs, and resource costs. Fixed service pricing yields higher welfare when that balance is reversed. Payments redistribute surplus. Welfare depends on how revenue, cost, and control align private computation choices with net social value.

\section*{AI Use Declaration}

AI tools assisted with code development, data analysis, and internal manuscript review. The authors determined the research questions and experimental design, conducted the experiments, and take responsibility for the final content.

\label{sec:main-end}
\bibliographystyle{ACM-Reference-Format}
\bibliography{references}

\appendix
\raggedbottom
\section{Proofs and Equilibrium Computation}
\label{app:proofs}

\subsection{Cost responsibility and finite menus}
\label{app:generalizations}

To establish local robustness, reallocate fixed total cost $T_\pi$ as $C_{M,\pi}=\beta T_\pi$ and $C_\pi=(1-\beta)T_\pi$. Every strict comparison supporting a baseline outcome is affine in $\beta$, say $a_j+b_j\beta>0$ with $a_j>0$. Choosing $0<\beta<1$ small enough that $|b_j|\beta<a_j/2$ preserves all signs and traded budgets. Total cost and value at each budget are unchanged, preserving welfare differences. Both suppliers bear positive real cost wherever $T_\pi(k)>0$. The argument applies to strict pooling and retail-pricing outcomes.

For Proposition~\ref{prop:finite-budgets}, cloud profit satisfies $f_1(k)=f_0(k)+s_k$. Optimality at $x\in K_0$ and $y\in K_1$ gives
\[
 0\leq f_0(x)-f_0(y)\leq s_y-s_x.
\]
If $y<x$, nondecreasing savings force equality, making both budgets optimal in both configurations. Choosing the smallest or largest optima rules out a decrease in either extreme. Strict downstream participation and positive initial best profit exclude rejection and exit. Arbitrary selections from tied optima need not be ordered. Under model control, fixed $E$ and lower cloud costs give $A_0\subseteq A_1$. Positive, strictly increasing $G_\pi$ makes the largest accepted budget uniquely preferable to smaller budgets and exit.

The savings condition is necessary for monotonicity uniformly over baseline profits. If $a<b$ but $s_a>s_b$, choose $f_0(b)-f_0(a)\in(0,s_a-s_b)$ and place all other profits sufficiently low. A common additive constant makes both best profits positive. The unique optimum falls from $b$ to $a$. A sufficiently large constant $E$ realizes this through $C_0=E-f_0$ and $C_1=C_0-s\geq0$. Symmetrically, with fixed $B,E$, lower model-firm cost expands its acceptance set under cloud control.

\subsection{Participation ties and exit}
\label{app:terminals}

The numerical solver retains pure subgame-perfect outcomes by backward induction. Let $O_a$ contain the outcomes available after action $a$, and let $u_i$ be the current decision maker's payoff. An outcome $o\in O_a$ is supported exactly when
\begin{equation}
 u_i(o)\geq\max_{b\ne a}\min_{z\in O_b}u_i(z).
 \label{eq:terminals}
\end{equation}
Necessity follows because a deviation to $b$ yields at least its minimum continuation payoff. For sufficiency, choose $o$ on path and a minimizing equilibrium continuation after each deviation. Recursion preserves acceptance or rejection at zero payoff, budget ties, and exit. All supported pure outcomes enter $\mathcal O_j$, and $\mathcal G$ contains their cross-pricing welfare differences.

\subsection{Configuration-dependent profiles}
\label{sec:identification}

When execution changes bills as well as costs, the cloud's sharing-margin change decomposes as
\begin{equation}
 m_1-m_0=(s_H-s_L)+\alpha(\Delta B_1-\Delta B_0),
 \label{eq:margin-decomp}
\end{equation}
where $m_\pi=\alpha\Delta B_\pi-\Delta C_\pi$. The identity separates cost savings from bill drift across configurations. Proposition~\ref{prop:finite-budgets} fixes $E$ across configurations. With configuration-dependent bills under sharing, increasing differences must be checked for the full profit $\alpha B_\pi-C_\pi$.

\subsection{Flexible retail and upstream terms}
\label{app:contract-benchmarks}

The following complete-information benchmarks allow prices or upstream payments to adjust. Acceptance at zero surplus or profit is retained as a supported continuation.

\subsubsection{Retail pricing}
For a representative risk-neutral user with $V_k>u_0$, consider the pointwise schedule $B_k=V_k-u_0$, or the limit of a common strict discount. It gives $\Delta B=\Delta V$. With $C_M=0$, $\Delta V>0$, and $\Delta C>0$, cloud-controlled sharing has margin $\alpha\Delta V-\Delta C$. The overprovision region $\Delta V<\Delta C<\alpha\Delta B$ is empty because $\alpha<1$.

The model-controlled participation channel remains possible. The cloud can accept the high budget when $C_\pi(k_H)\leq\alpha(V_H-u_0)$. Under the remaining conditions of Proposition~\ref{prop:retail-pooling-admission}, its welfare change is still $\Delta V-\Delta C_1+s_L$. These conclusions condition on the pointwise schedule, not an equilibrium of unrestricted menu pricing, where prices at other budgets can also steer the cloud.

\subsubsection{Upstream payments}
Suppose the model firm proposes the contract and every considered budget satisfies user participation. Under model control, let $k^*$ maximize $B_k-C(k)-C_M(k)$. If this maximum is nonnegative, a capacity fee $L=C(k^*)$ supports $k^*$ at zero cloud profit. The firm's best attainable profit, including exit, is
\[
 \max\bigl\{0,\max_{k\in\Kset}[B_k-C(k)-C_M(k)]\bigr\}.
\]
Cloud participation bounds profit by this expression. Any cloud-accepted alternative has cost at most the fee and cannot give higher firm profit than the integrated maximizer. The cost-covering fee therefore attains the bound. Combined with pointwise retail extraction, the firm maximizes social welfare.

For cloud control, specialize to the two-budget menu with $C_M=0$, $0<B_L<B_H$, and $0\leq C_L<C_H$. Supporting the high budget through sharing requires both an incentive and participation,
\[
 \alpha\geq\alpha_H
 =\max\left\{\frac{\Delta C}{\Delta B},\frac{C_H}{B_H}\right\}.
\]
When $\alpha_H<1$, the cloud's rent at that threshold is $\alpha_H B_H-C_H$. These two inequalities follow by comparing cloud profits across budgets and with exit.

With a signed fixed transfer, the two-part payment $E(k)=F+\alpha B_k$ can remove this rent \cite{oi1971}. When $\alpha_C=\Delta C/\Delta B\in(0,1)$, setting $\alpha=\alpha_C$ and $F=C_L-\alpha_C B_L$ makes cloud profit zero at both budgets. Either user-accepted budget with $B_k\geq C_k$ is supportable. The firm's best attainable profit becomes $\max\{0,\max_k(B_k-C_k)\}$. This follows by substituting the transfer into the two cloud payoffs. Flexible transfers thus recover the integrated supplier objective as an attainable benchmark.

\section{Experimental and Price Calibration}
\label{app:calibration}
\label{app:service-measurements}

\subsection{Deployments and service accounting}
\label{app:admission}

Qwen3.8-27B runs with thinking low or off. The other deployments are Qwen3.5-9B and Nemotron-3-Nano-30B-A3B with default thinking, and non-thinking Qwen2.5-Coder-7B-Instruct. Table~\ref{tab:pilots} records their final calibration limits and capped candidates. Each pilot contains 128 candidates. Settings were frozen before correctness labeling. The non-thinking Qwen3.8 control inherits the thinking cap and context.

\begin{table}[t]
 \caption{Frozen generation settings.}
 \label{tab:pilots}
 \centering\small
 \begin{tabular}{@{}lrrr@{}}
 \toprule
 Deployment & \shortstack{Output /\\context cap} & \shortstack{Call limit\\(s)} & Capped \\
 \midrule
 Qwen3.8 on & 24,576 / 32,768 & 3,600 & 14/128 \\
 Qwen3.8 off & 24,576 / 32,768 & 3,600 & 5/128 \\
 Qwen3.5 & 24,576 / 32,768 & 3,600 & 44/128 \\
 Nemotron & 24,576 / 32,768 & 3,600 & 43/128 \\
 Coder & 4,096 / 16,384 & 1,800 & 0/128 \\
 \bottomrule
 \end{tabular}
\end{table}

The service window runs from cohort dispatch through the last online selection or abstention. CPU counters cover the inference server, selector, and evaluator over that window. Offline correctness audit and idle harness roles are excluded. Startup-inclusive sensitivity adds the full GPU step and cumulative CPU of those service roles. Each repeated budget/configuration comparison uses one GPU allocation. Repetitions can differ in GPU model and share nodes. All service calls completed within their limits, with a maximum of 3{,}503 s. Service results are descriptive on the fixed tasks and observed allocations.

\subsection{Matched pooling comparisons}
\label{app:matched-pooling}

Table~\ref{tab:service-h1} gives all 18 within-repetition comparisons. Here $d_0$ and $d_1$ are the increases in GPU seconds per request from one to eight candidates under serial and pooled execution. Their difference measures the reduction in this resource increment. Each row uses four cells from one GPU allocation. H100 denotes H100 HBM3, and NVL denotes H100 NVL. Contrasts are calculated before rounding.


\begin{table*}[t]
\caption{Online-service cost contrasts (GPU s/request).}
\label{tab:service-h1}
\centering\small
\begin{tabular}{lllrrrrr}
\toprule
Deployment & Pooling & Repeat & GPU & $d_0$ & $d_1$ & $d_0-d_1$ & $(d_0-d_1)/d_0$ \\
\midrule
Coder & batch16 & r0 & NVL & 4.8 & 3.6 & 1.16 & 24.3\% \\
Coder & batch16 & r1 & H200 & 3.9 & 3.1 & 0.74 & 19.3\% \\
Coder & batch16 & r2 & NVL & 4.1 & 3.4 & 0.61 & 15.2\% \\
Nemotron & batch16 & r0 & H200 & 66.3 & 26.1 & 40.19 & 60.6\% \\
Nemotron & batch16 & r1 & H200 & 66.7 & 26.4 & 40.24 & 60.4\% \\
Nemotron & batch16 & r2 & NVL & 91.4 & 38.1 & 53.35 & 58.4\% \\
Qwen3.5 thinking & batch16 & r0 & H100 & 40.7 & 33.0 & 7.63 & 18.8\% \\
Qwen3.5 thinking & batch16 & r1 & NVL & 47.7 & 41.8 & 5.91 & 12.4\% \\
Qwen3.5 thinking & batch16 & r2 & H100 & 42.0 & 33.0 & 9.06 & 21.6\% \\
Qwen3.8 non-thinking & batch16 & r0 & H100 & 76.1 & 26.9 & 49.20 & 64.6\% \\
Qwen3.8 non-thinking & batch16 & r1 & NVL & 89.2 & 21.3 & 67.86 & 76.1\% \\
Qwen3.8 non-thinking & batch16 & r2 & H100 & 96.8 & 26.0 & 70.82 & 73.1\% \\
Qwen3.8 non-thinking & batch64 & r0 & H100 & 76.1 & 36.8 & 39.31 & 51.7\% \\
Qwen3.8 non-thinking & batch64 & r1 & NVL & 89.2 & 28.4 & 60.82 & 68.2\% \\
Qwen3.8 non-thinking & batch64 & r2 & H100 & 96.8 & 33.1 & 63.69 & 65.8\% \\
Qwen3.8 thinking & batch16 & r0 & H100 & 86.6 & 60.5 & 26.08 & 30.1\% \\
Qwen3.8 thinking & batch16 & r1 & NVL & 97.2 & 50.3 & 46.95 & 48.3\% \\
Qwen3.8 thinking & batch16 & r2 & H200 & 70.7 & 21.8 & 48.87 & 69.1\% \\
\bottomrule
\end{tabular}
\end{table*}

All 18 contrasts are positive. The service window includes online selection, so the measured reduction covers generation and answer delivery. These resource increments contribute to $\Delta C_\pi$ after cost conversion. Changes in token bills enter the sharing margin separately through equation~\eqref{eq:margin-decomp}.

\subsection{Tariff anchors and reference prices}
\label{app:tariffs}

The 8 October catalog anchors quote Qwen3.5-9B at 0.10/0.15 USD per million input/output tokens, with 0.17/0.25 as an alternative. Qwen3.8-27B has input and output list rates of 0.20 and 2.50 USD. Its promotional rates are 0.15 and 1.875, and the alternative rates are 0.40 and 3.00. Only the first provider's catalog response is archived. Exact hosted-weight identity is unverified, so these rates serve as calibration anchors. Reasoning and unselected output use the output rate. Coder and Nemotron have no durable matching tariff anchor.

The controlled comparison sets $B_A(k)=b+p_hH_k$, where $b$ averages input and delivered-output charges across both budgets and three repetitions. The reference price $P(1/2)$ averages those metered bills over the same six cells. A separate sensitivity uses each cell's complete observed token bill. All \rtMainProfiles{} main-tariff budget profiles agree between constructions, with one promotional-tariff participation difference across seven deployment/configuration combinations and three repetitions. The capacity fee remains $P(1/2)/2$ when $w$ varies.

\subsection{Resource-cost conversion}
\label{app:anchors}

Per-request cost divides cohort GPU, CPU, and energy costs by 106. The central rates are 2.4 USD per busy GPU hour for energy-exclusive capital and hosting, 0.03 USD per CPU core hour, and 0.15 USD per kWh. Energy is estimated at the midpoint of 35--90\% of rated GPU power during generation and 10--15\% otherwise. Sensitivity varies the cell energy endpoints and includes startup. Rates are declared scenarios, and user values remain parameters.

The capital-and-hosting model uses a 285,000 USD eight-GPU server, 18\% networking/storage overhead, five-year amortization at 9.33\%, annual operations at 3.5\% of capital, and hosting at 150 USD per kW-month for 10.2 kW of server power with a 1.1 multiplier. The central utilization assumption is 0.7. Utilization from 0.9 to 0.5 gives the 1.9--3.4 USD range. Public prices provide context, not observed supplier costs. The full assumptions, source dates, anonymous records, and CPU-overlap sensitivity accompany the calculations.

\subsection{Protocol timeline}
\label{app:timeline}

The cohort was fixed on 9 October. The 10 October protocol specified pooling hypotheses, paired bootstrap, a cost-spread pause, and configuration comparisons. The initial pipeline included all-candidate judging. Later service measurements use online selection and exclude offline audit. The 21 October amendment admitted hardware variants, and subsequent reruns changed selected call limits. The cost grid expanded on 24 October, and the service measurements and retail-pricing counterfactual were analyzed on 25 October. The integrated-supplier benchmark was added on 26 October. The original primary bootstrap was withheld by its pause rule. Separate pipeline diagnostics remain in the supplemental material and are not intervals for the service comparisons.

\section{Numerical Checks and Price Thresholds}
\label{app:retail-details}

The supplemental material supplied with this submission retains the service inputs, exact calculations, and historical analyses with their original scopes.

\subsection{Service profiles and welfare thresholds}

Table~\ref{tab:retail-profiles} gives batch16 delivery rates and resource costs for all five deployments. Delivery rates average the three repetitions, while each cost range spans their separate allocations at the central scenario. The profiles supply the quality and cost inputs for the budget comparison.

\begin{table}[t]
\caption{Batch16 service profiles.}
\label{tab:retail-profiles}
\centering\small
\begin{tabular}{@{}lcrrr@{}}
\toprule
Deployment & $k$ & $q^+$ (\%) & $q^-$ (\%) & \shortstack{Cost\\(USD/1k)} \\
\midrule
Coder & 1 & 18.6 & 7.9 & 0.72--0.76 \\
 & 8 & 26.1 & 14.8 & 2.87--3.16 \\
Nemotron & 1 & 52.2 & 4.7 & 8.97--12.79 \\
 & 8 & 67.0 & 8.8 & 26.93--38.67 \\
Qwen3.5 & 1 & 39.3 & 1.6 & 10.75--10.88 \\
 & 8 & 51.6 & 3.5 & 33.48--39.11 \\
Qwen3.8 off & 1 & 43.7 & 6.6 & 18.67--20.26 \\
 & 8 & 65.1 & 13.8 & 34.75--37.15 \\
Qwen3.8 on & 1 & 57.9 & 4.7 & 19.56--25.32 \\
 & 8 & 78.3 & 6.6 & 34.54--66.94 \\
\bottomrule
\end{tabular}
\end{table}

Table~\ref{tab:retail-thresholds} records the coefficients plotted in Figure~\ref{fig:threshold-coefficients}. For each repetition, the welfare difference from one to eight candidates is $\Delta q^+(v-a-bd-ch)$, with $h$ measured in USD per hour of waiting. A larger cost coefficient $a$ raises the required answer value even when error loss and waiting cost are zero. The coefficients $b$ and $c$ capture the extra error loss and waiting cost per additional correct delivery.

\begin{table}[t]
\caption{Welfare-threshold coefficients by deployment.}
\label{tab:retail-thresholds}
\centering\small
\begin{tabular}{@{}llrrr@{}}
\toprule
Deployment & Repeat & $a$ & $b$ & $c$ \\
\midrule
Coder & $r_0$ & 0.0287 & 0.8889 & 0.6480 \\
 & $r_1$ & 0.0447 & 1.8000 & 1.0229 \\
 & $r_2$ & 0.0247 & 0.5000 & 0.5446 \\
Nemotron & $r_0$ & 0.1269 & 0.3333 & 3.3068 \\
 & $r_1$ & 0.1207 & 0.1250 & 3.1162 \\
 & $r_2$ & 0.1714 & 0.3750 & 4.5158 \\
Qwen3.5 & $r_0$ & 0.1724 & 0.2143 & 4.4047 \\
 & $r_1$ & 0.3340 & 0.1111 & 8.5517 \\
 & $r_2$ & 0.1505 & 0.1250 & 3.8625 \\
Qwen3.8 off & $r_0$ & 0.1031 & 0.4211 & 2.9670 \\
 & $r_1$ & 0.0668 & 0.2174 & 1.8366 \\
 & $r_2$ & 0.0728 & 0.3846 & 2.0970 \\
Qwen3.8 on & $r_0$ & 0.2101 & 0.1429 & 5.0805 \\
 & $r_1$ & 0.1722 & 0.0952 & 4.2478 \\
 & $r_2$ & 0.0690 & 0.0435 & 1.7295 \\
\bottomrule
\end{tabular}
\end{table}

The budget comparison also retains the purchase conditions in Section~\ref{sec:retail-welfare-results}. Coder and Nemotron provide social thresholds and the required-price conditions below. The three quoted deployments additionally support the numerical retail-contract comparisons.

\subsection{Supplier participation}
\label{app:supplier-participation}

Table~\ref{tab:retail-decisions} gives all sharing outcomes under strict user purchase. Subscripts $M,C$ identify control, and $\varnothing$ denotes no service.

\begin{table}[t]
\caption{Budget choices (batch16, $\alpha=0.5$, $\rho=2.4$).}
\label{tab:retail-decisions}
\centering\small
\begin{tabular}{@{}llcccc@{}}
\toprule
Deployment & Repeat & $A_M$ & $D_M$ & $A_C$ & $D_C$ \\
\midrule
Qwen3.5 & $r_0$ & $\varnothing$ & $\varnothing$ & $\varnothing$ & $\varnothing$ \\
 & $r_1$ & $\varnothing$ & $\varnothing$ & $\varnothing$ & $\varnothing$ \\
 & $r_2$ & $\varnothing$ & $\varnothing$ & $\varnothing$ & $\varnothing$ \\
Qwen3.8 off & $r_0$ & $8$ & $1$ & $8$ & $1$ \\
 & $r_1$ & $8$ & $1$ & $8$ & $1$ \\
 & $r_2$ & $8$ & $1$ & $8$ & $1$ \\
Qwen3.8 on & $r_0$ & $8$ & $1$ & $8$ & $1$ \\
 & $r_1$ & $8$ & $1$ & $8$ & $1$ \\
 & $r_2$ & $8$ & $\{1,8\}$ & $8$ & $1$ \\
\bottomrule
\end{tabular}
\end{table}

Table~\ref{tab:retail-payoffs} supplies the bills, costs, and profits underlying Figure~\ref{fig:supplier-profits}, in USD per 1,000 requests. It uses Qwen3.8 non-thinking, batch16, $r_0$, with $\alpha=1/2$ and $C_M=0$. These are payoffs conditional on executing each proposed service. Negative cloud profit explains rejection. The fixed price averages all three repetitions.

\begin{table}[t]
\caption{Qwen3.8 off payoffs ($r_0$, batch16, $\alpha=0.5$), USD/1,000 requests.}
\label{tab:retail-payoffs}
\centering\small
\begin{tabular}{@{}lrrrrr@{}}
\toprule
Pricing & $k$ & Bill & Cost & $\Pi_M$ & $\Pi_C$ \\
\midrule
Metered & 1 & 14.27 & 18.67 & 7.13 & $-11.54$ \\
Metered & 8 & 98.55 & 37.15 & 49.27 & 12.13 \\
Fixed & 1 & 57.77 & 18.67 & 28.89 & 10.22 \\
Fixed & 8 & 57.77 & 37.15 & 28.89 & $-8.26$ \\
\bottomrule
\end{tabular}
\end{table}

\begin{table}[t]
\caption{Price floors for Coder and Nemotron.}
\label{tab:retail-unquoted}
\centering\small
\begin{tabular}{@{}llrrrr@{}}
\toprule
Deployment & Repeat & $\gamma_1$ & $\eta_1$ & $\gamma_8$ & $\eta_8$ \\
\midrule
Coder & $r_0$ & 2.4887 & 1.8197 & 1.3252 & 1.7682 \\
 & $r_1$ & 2.6445 & 1.8329 & 1.2223 & 1.7944 \\
 & $r_2$ & 2.5385 & 1.8182 & 1.2904 & 1.7745 \\
Nemotron & $r_0$ & 0.6328 & 0.0377 & 0.2367 & 0.0375 \\
 & $r_1$ & 0.6386 & 0.0373 & 0.2360 & 0.0368 \\
 & $r_2$ & 0.8487 & 0.0354 & 0.3307 & 0.0365 \\
\bottomrule
\end{tabular}
\end{table}

Under strict user purchase, the capacity-fee outcomes in Table~\ref{tab:supplier-structure} hold whenever
\[
 \max\{B_A(k_L),C_L\}<L<\min\{P,C_H\}.
\]
The model firm rejects the low metered budget, and the cloud rejects the high one. Both profit from the low fixed-price service. The common interval contains \rtIntegrationFeeLo--\rtIntegrationFeeHi{} USD per 1,000 requests, rounded inward, and the reference fee of \rtIntegrationFee{} USD. Exact boundary ties are retained.

\subsection{Models without quoted tariffs}

Coder and Nemotron identify resource and quality requirements without a matching retail quote. Let $I_k,N_k$ be mean input and total output tokens at budget $k$. The complete bill is $B_k=(p_iI_k+p_oN_k)/10^6$, including delivered, reasoning, and unselected output at rates per million tokens.

For $N_k>0$, strict cloud participation requires
\begin{equation}
 p_o>\frac{\gamma_k}{\alpha}-\eta_kp_i,\qquad
 \gamma_k=\frac{10^6C_k}{N_k},\quad \eta_k=\frac{I_k}{N_k}.
 \label{eq:unquoted-price-floor}
\end{equation}
Strict user purchase also requires $p_o<[10^6(V_k-u_0)-p_iI_k]/N_k$. Table~\ref{tab:retail-unquoted} gives these price-bound coefficients at each budget. Social provision compares each budget with exit through $v>[C_k+dq_k^-+h\tau_k/3600+u_0]/q_k^+$ when $q_k^+>0$, alongside the welfare threshold in Figure~\ref{fig:threshold-coefficients}.

\subsection{Pooling admission scans}
\label{app:pooling-admission}

The cloud-side inequalities for the separate pooling-induced participation switch in Proposition~\ref{prop:retail-pooling-admission} have no admissible share in any of the \admPreRows{} pipeline and \svNewPreRows{} service scenario rows under complete non-delivered-token billing. Both counts include rows for the minimum-cost pooled menu and reuse measurements under different price and cost assumptions.

For serial bills and costs, $C_0(k_H)/B_0(k_H)$ is below $C_0(k_L)/B_0(k_L)$ in every service row. A share sufficient to cover the smaller service therefore covers the larger one, excluding the initial acceptance pattern required by Proposition~\ref{prop:retail-pooling-admission}.

\end{document}